\documentclass[11pt,a4paper]{article}

\usepackage[margin=1in]{geometry}
\usepackage[T1]{fontenc}
\usepackage{lmodern}
\usepackage{amsmath,amssymb,amsthm}
\usepackage{graphicx}
\usepackage{xspace}
\usepackage{float}
\usepackage{pifont}
\usepackage{tabularx}
\usepackage{multirow}
\usepackage{subcaption}
\usepackage{xcolor}
\usepackage[hyphens]{url}
\usepackage[hidelinks]{hyperref}
\usepackage{aliascnt}
\usepackage[capitalise,noabbrev]{cleveref}

\newtheorem{theorem}{Theorem}
\newaliascnt{lemma}{theorem}
\newtheorem{lemma}[lemma]{Lemma}
\aliascntresetthe{lemma}
\newaliascnt{conjecture}{theorem}
\newtheorem{conjecture}[conjecture]{Conjecture}
\aliascntresetthe{conjecture}
\newaliascnt{observation}{theorem}
\newtheorem{observation}[observation]{Observation}
\aliascntresetthe{observation}
\theoremstyle{definition}
\newaliascnt{definition}{theorem}
\newtheorem{definition}[definition]{Definition}
\aliascntresetthe{definition}
\crefname{conjecture}{Conjecture}{Conjectures}
\Crefname{conjecture}{Conjecture}{Conjectures}
\crefname{observation}{Observation}{Observations}
\Crefname{observation}{Observation}{Observations}
\DeclareRobustCommand{\rev}[1]{#1}

\title{Connectivity Bounds for GHZ Graphs}
\author{L. Sunil Chandran\thanks{Indian Institute of Science, Bengaluru, India.}
\and Rishikesh Gajjala\thanks{Center for Quantum and Topological Systems, New York University Abu Dhabi, United Arab Emirates.}}
\date{}
\hypersetup{
  pdftitle={Connectivity Bounds for GHZ Graphs},
  pdfauthor={L. Sunil Chandran, Rishikesh Gajjala},
  pdfsubject={Author version accepted at FSTTCS 2026},
  pdfkeywords={GHZ graphs, graph connectivity, perfect matchings, quantum optics}
}

\begin{document}
\maketitle

\begin{abstract}
Greenberger--Horne--Zeilinger (GHZ) states are multipartite entangled quantum states that play a central role in quantum information theory, with applications in quantum communication, cryptography, and distributed quantum computing. A major experimental goal is to create high-dimensional GHZ states using multiple entangled parties. Krenn, Gu and Zeilinger discovered a surprising bridge between these quantum-optical experiments and graph theory: certain photonic experiments can be represented by edge-coloured, edge-weighted graphs, and the dimension of the produced GHZ state equals the number of feasible monochromatic inherited vertex colourings.

{This led Krenn and Gu to conjecture that every GHZ graph with more than four vertices has dimension at most $2$. A counterexample would give a new high-dimensional multipartite GHZ construction, while a proof would clarify fundamental resource limitations in quantum optics. It is known that the conjecture holds for graphs of vertex connectivity at most $2$ [MFCS 2024], and that simple GHZ graphs with $n>4$ vertices have dimension less than $n/\sqrt{2}$ [Quantum 2024]. However, no general upper bound depending only on vertex connectivity was known.} We prove that every simple connected GHZ graph $G$ with at least four vertices has dimension at most $
    2^{\kappa(G)-1}+\binom{\kappa(G)}{2},$
where $\kappa(G)$ is the vertex connectivity of $G$. We also prove new bounds in terms of edge connectivity and maximum degree. 

\end{abstract}

\newcommand{\cmark}{\ding{51}} % check mark
\newcommand{\xmark}{\ding{55}} % x mark

%%%%%%%%%%%%%%%%%%%%%%%%%%%%%%%%%%%%%%%%%%%%%%%%%%%%%%

\section{Introduction}
\label{sec:intro}
Quantum entanglement is one of the most remarkable phenomena predicted by quantum physics: measurements on spatially separated particles can exhibit correlations that cannot be explained by local realistic theories~\cite{PhysRev.47.777}. Greenberger, Horne and Zeilinger asked what happens when more than two particles are entangled~\cite{Greenberger1989}. The resulting Greenberger--Horne--Zeilinger (GHZ) states, such as
\[
    |GHZ_{3,2}\rangle=\frac{1}{\sqrt{2}}\left(|000\rangle+|111\rangle\right),
\]
were experimentally observed and used to rule out local-realistic descriptions~\cite{PhysRevLett.82.1345,Pan2000}. Today, GHZ states are central objects in quantum information theory, with applications in quantum secret sharing~\cite{PhysRevA.59.1829}.

{Multipartite entanglement is closely connected to the experimental work for which Anton Zeilinger received the 2022 Nobel Prize in Physics~\cite{nobel}. Increasing both the number of particles and the local dimension of GHZ states is important for foundational tests of quantum mechanics as well as for practical quantum-information tasks. Photonic technology is one of the main experimental platforms for this goal~\cite{Gu2020}. Motivated by this, Krenn, Gu, Zeilinger and co-authors developed graph-based descriptions and automated design methods for quantum-optical experiments~\cite{Quantum_graphs,Quantum_graphs_2,Quantum_graphs_3,krenn2019questions}.}

{The key insight is that a quantum-optical experiment with probabilistic photon sources and linear optics can be encoded by an edge-coloured, edge-weighted graph. Edges represent photon-pair sources, edge colours represent photonic modes, and edge weights represent complex amplitudes. A successful detection event in which every detector receives exactly one photon corresponds to a perfect matching. Thus the quantum question
\[
    \text{``Can one create a high-dimensional GHZ state by such an experiment?''}
\]
becomes a graph-theoretic question about cancellations among weighted perfect matchings.}

{Despite substantial experimental and computational effort, no construction is known that produces a GHZ state of dimension $d>2$ with more than $n=4$ photons, perfect quality, finite count rates, and no additional resources~\cite{krenn2019questions}. This led Krenn and Gu to conjecture that such a construction is impossible. In graph-theoretic terms, their conjecture says that every GHZ graph on more than four vertices has dimension at most $2$; see \Cref{krenn_conjecture}. The conjecture has been widely publicised among graph theorists and quantum-information researchers~\cite{mixon_website,krenn_website}. A counterexample would yield a new high-dimensional multipartite GHZ experiment, while a proof would reveal structural limitations of this photonic model.}

We now state the graph-theoretic formulation. For a high-level explanation of the physical encoding, we refer the reader to Section~2.2 of~\cite{CG1}; for the full quantum-optical model, see~\cite{Quantum_graphs,Quantum_graphs_2,Quantum_graphs_3,krenn2019questions}.

\subsection{Problem definition}
We consider finite loopless undirected multigraphs; parallel edges are allowed, but loops are not.
For a graph $G$, let $V(G)$ and $E(G)$ denote its vertex and edge sets.
For a connected graph $G$, the \emph{edge connectivity} $\lambda(G)$ is the minimum number
of edges whose deletion disconnects $G$, where parallel edges are counted with multiplicity.
We set $\lambda(K_1)=0$.
If $G$ is connected and not complete, the \emph{vertex connectivity} $\kappa(G)$ is the
minimum size of a vertex set $T\subseteq V(G)$ such that $G-T$ is disconnected; for complete
graphs we use the standard convention $\kappa(K_n)=n-1$.
For disconnected graphs, we set $\lambda(G)=\kappa(G)=0$.
We write $\Delta_v(G)$ and $\delta_v(G)$ for the maximum and minimum number of neighbours
of a vertex, and $\Delta_e(G)$ and $\delta_e(G)$ for the corresponding maximum and minimum
numbers of incident edges. For simple graphs, these coincide with the usual maximum and
minimum degrees, denoted $\Delta(G)$ and $\delta(G)$. An edge is \emph{redundant} if it
belongs to no perfect matching. Removing all redundant edges from $G$ gives its
\emph{matching-covered graph}, denoted $mcg(G)$.

An edge colouring $c:E(G)\to\mathbb{N}$ assigns a single colour to each edge; 
physically, the colour represents a photonic mode number. All definitions in this 
paper use this ordinary edge-colouring model. In particular, an edge has one colour, 
not one colour at each endpoint. The more general bichromatic, or half-edge, model 
assigns colours separately to the two ends of an edge and is discussed only when 
comparing with prior work. An edge weight function 
$w:E(G)\to\mathbb{C}\setminus\{0\}$ assigns a non-zero complex amplitude to each 
edge. A graph $G$ together with $c$ and $w$ is denoted by $G_c^w$ and is called an 
edge-coloured, edge-weighted graph. For $e\in E(G)$, $c(e)$ and $w(e)$ denote the 
colour and the weight of $e$, respectively.

\begin{figure}[t!]
    \centering   
    \begin{minipage}{0.9\textwidth}
\centering    
{\includegraphics[width=63mm]{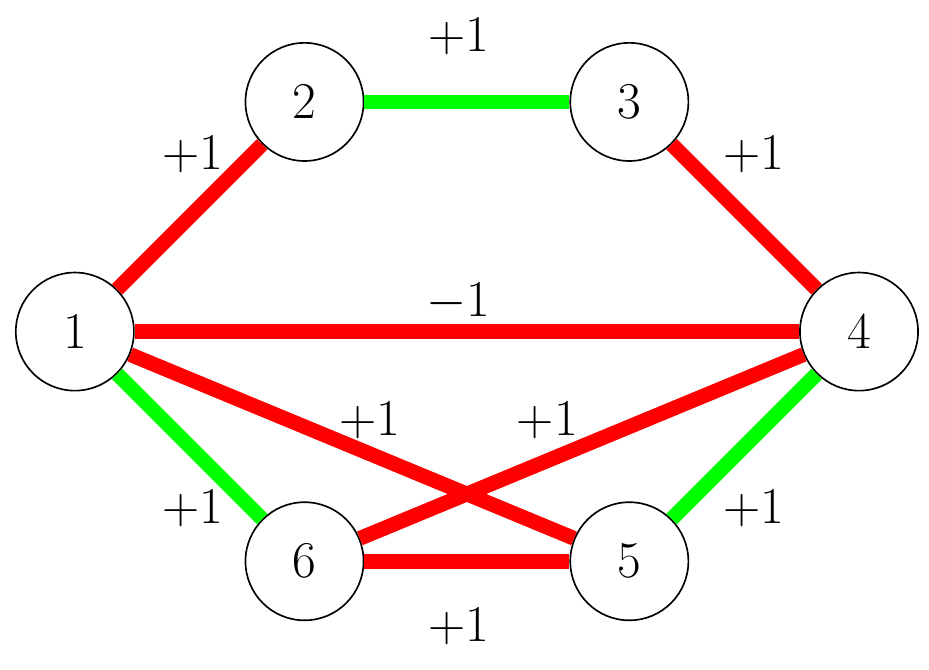}}
\caption{An edge-coloured edge-weighted graph. The edges $(1,6),(4,5),(2,3)$ are of colour green (mode number $1$). Remaining edges are of colour red (mode number $0$).}
\label{fig:main_example}
    \end{minipage}
   \hspace*{1cm}
\end{figure}

\begin{figure}[t!]
    \centering
\centering    
\begin{subfigure}[b]{0.23\textwidth}
         \centering
         \fbox{\includegraphics[width=\dimexpr\linewidth-2\fboxsep-2\fboxrule\relax]{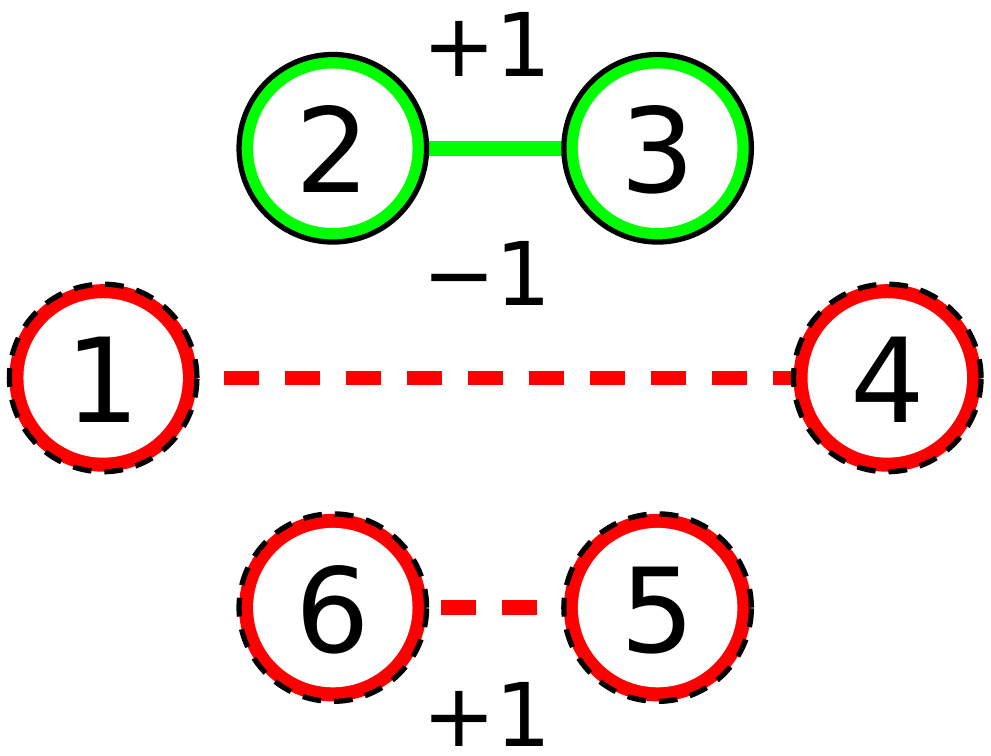}}
         \caption{$|011000\rangle$}
         \label{fig:pm1}
\end{subfigure}
\hfill
\begin{subfigure}[b]{0.23\textwidth}
         \centering
         \fbox{\includegraphics[width=\dimexpr\linewidth-2\fboxsep-2\fboxrule\relax]{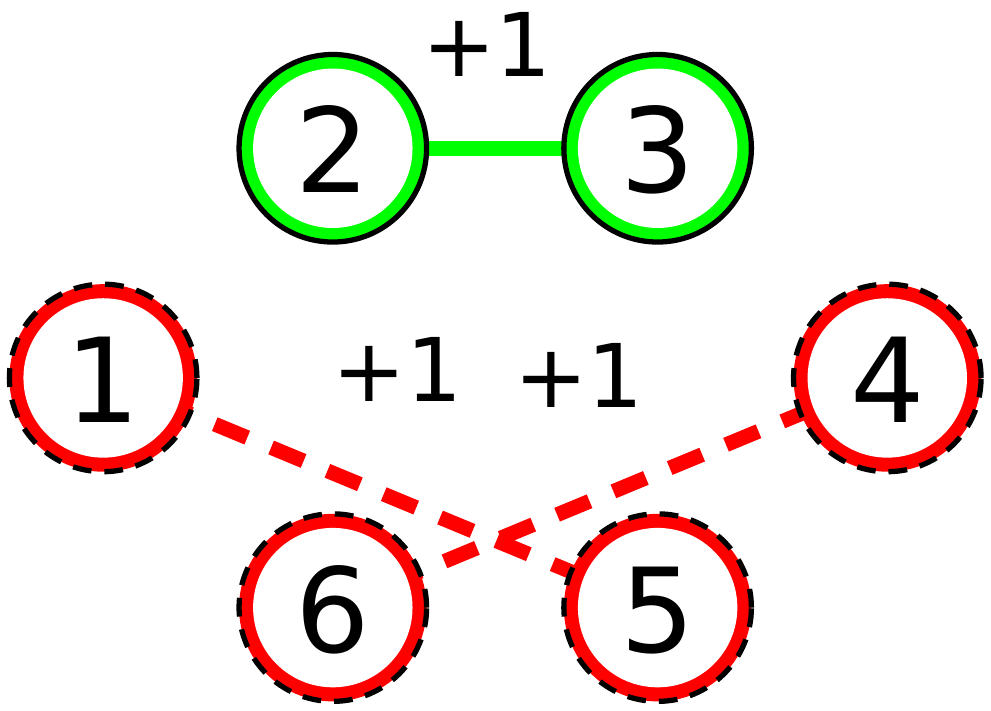}}
         \caption{$|011000\rangle$}
         \label{fig:pm2}
\end{subfigure}
\hfill
\begin{subfigure}[b]{0.23\textwidth}
         \centering
         \fbox{\includegraphics[width=\dimexpr\linewidth-2\fboxsep-2\fboxrule\relax]{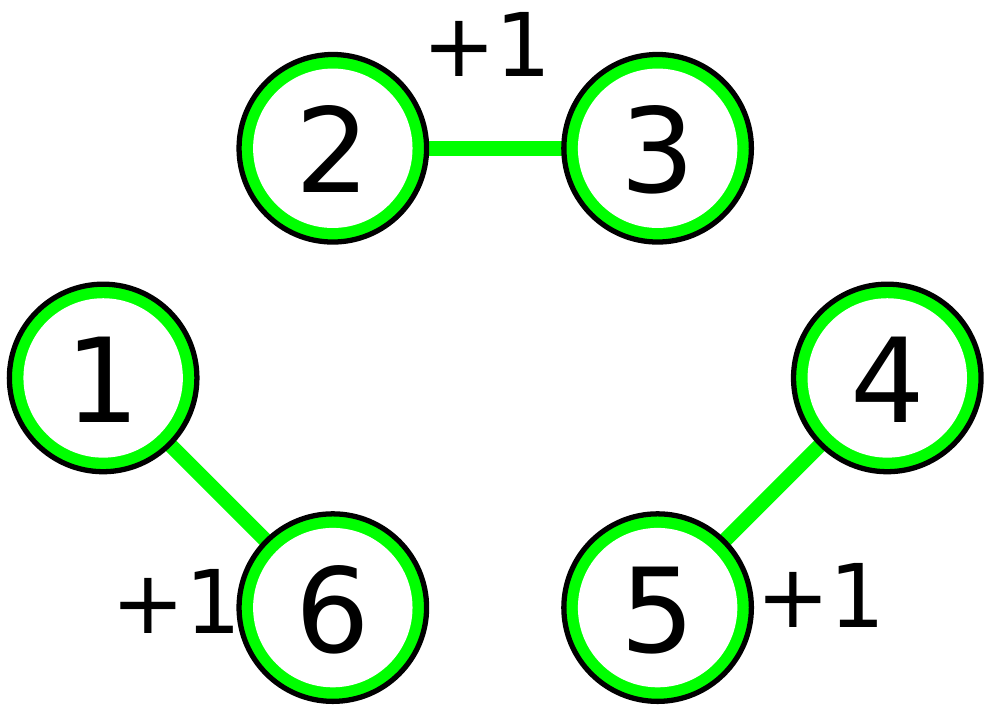}}
         \caption{$|111111\rangle$}
         \label{fig:pm3}
\end{subfigure}
\hfill
\begin{subfigure}[b]{0.23\textwidth}
         \centering
         \fbox{\includegraphics[width=\dimexpr\linewidth-2\fboxsep-2\fboxrule\relax]{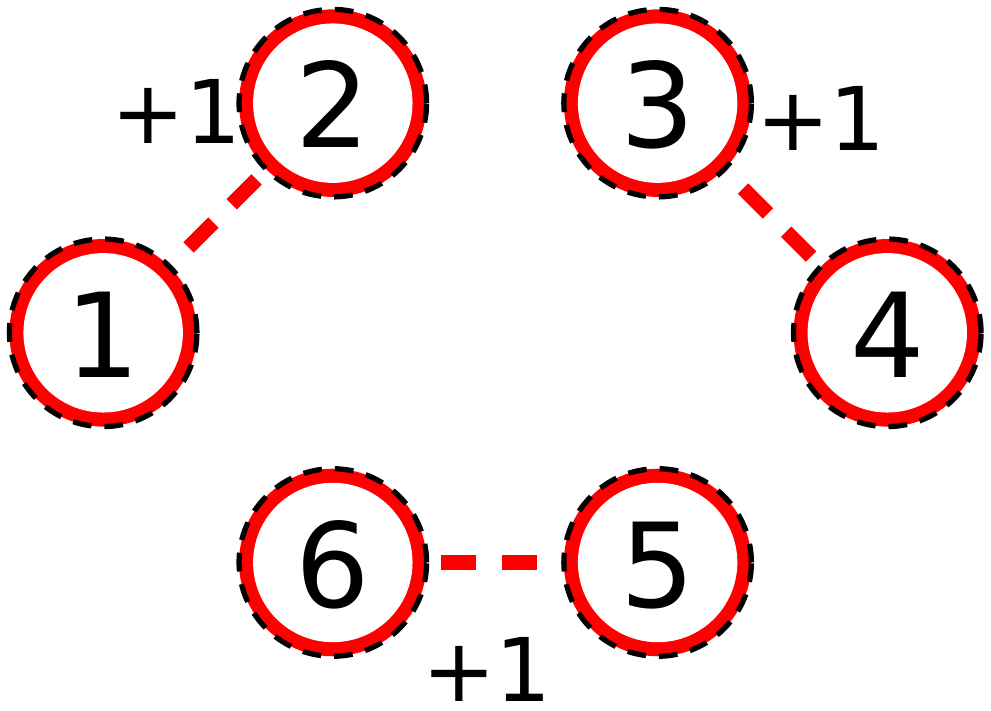}}
         \caption{$|000000\rangle$}
         \label{fig:pm4}
\end{subfigure}

\caption{Inherited vertex colourings of the perfect matchings of \cref{fig:main_example}. Red edges and vertex outlines are dashed; green edges and outlines are solid.}
\label{fig:perfect_matchings}
\end{figure}

{Every perfect matching $M$ of $G_c^w$ induces a vertex colouring $vc_M:V(G)\to\mathbb{N}$ as follows: each vertex receives the colour of the unique edge of $M$ incident with it. In the quantum-optical interpretation, this records the mode detected at each photon detector. We call $vc_M$ an \emph{inherited vertex colouring}. It is \emph{monochromatic} if all vertices receive the same colour. For a vertex colouring $vc$, let $\mathcal{M}_{vc}$ be the set of perfect matchings inducing $vc$. If $\mathcal{M}_{vc}\neq\emptyset$, then $vc$ is \emph{feasible}; otherwise it is \emph{infeasible}.}

\begin{definition}
{The weight of a perfect matching $M$ is
\[
    w(M)=\prod_{e\in M}w(e).
\]}
\end{definition}

\begin{definition}
{The weight of an inherited vertex colouring $vc$ is
\[
    w(vc)=\sum_{M\in\mathcal{M}_{vc}}w(M).
\]
If $\mathcal{M}_{vc}=\emptyset$, we set $w(vc)=0$.}
\end{definition}

\begin{definition}
{An edge-coloured, edge-weighted graph $G_c^w$ is a \emph{GHZ graph} if:
\begin{enumerate}
    \item every feasible monochromatic inherited vertex colouring has weight $1$; and
    \item every non-monochromatic inherited vertex colouring has weight $0$.
\end{enumerate}}
\end{definition}

{Thus, in a GHZ graph, all unwanted non-monochromatic detection patterns cancel by destructive interference, while the desired monochromatic patterns remain with weight $1$. Figure~\ref{fig:perfect_matchings} explains this definition for the graph in Figure~\ref{fig:main_example}. The perfect matchings in Figures~\ref{fig:pm1} and~\ref{fig:pm2} induce the same non-monochromatic inherited vertex colouring $|011000\rangle$, and their weights cancel. The perfect matchings in Figures~\ref{fig:pm3} and~\ref{fig:pm4} induce the monochromatic inherited vertex colourings $|111111\rangle$ and $|000000\rangle$, each with weight $1$. Hence Figure~\ref{fig:main_example} is a GHZ graph of dimension $2$.}

\begin{definition}
The dimension of a GHZ graph $G^w_c$, denoted $\mu(G,c,w)$, is the number of
feasible monochromatic inherited vertex colourings.
\end{definition}

If $G^w_c$ is not a GHZ graph, then $\mu(G,c,w)$ is not defined.  For a
loopless multigraph $G$, let
\[
WC(G)=\{(w,c):G^w_c \text{ is a GHZ graph}\}.
\]
The matching index of $G$ is
\[
\mu(G)=\max_{(w,c)\in WC(G)} \mu(G,c,w),
\]
where the maximum over the empty set is defined to be $0$.

This maximum is taken over the fixed graph $G$ itself: no parallel copies of
edges are added unless they are already present in $G$.  It is easy to see that
\[
\mu(G)=\mu(\operatorname{mcg}(G)).
\]

For a simple graph $H$, if one instead maximizes over all loopless multigraphs
whose skeleton is $H$, we denote the resulting quantity by
$\mu_{\mathrm{skel}}(H)$.  Thus Figure~\ref{fig:k2} shows that
$\mu_{\mathrm{skel}}(K_2)=\infty$.  In the rest of the paper, and in particular
in Theorems~\ref{thm:edge_connectivity}, \ref{thm:vertex_connectivity}, and
\ref{thm:six_eight_vertices}, $\mu(G)$ always denotes the fixed-graph quantity above.

\begin{figure}[t!]
    \centering
\centering    

\begin{subfigure}[b]{0.4\textwidth}
         \centering
         {\includegraphics[width=40mm]{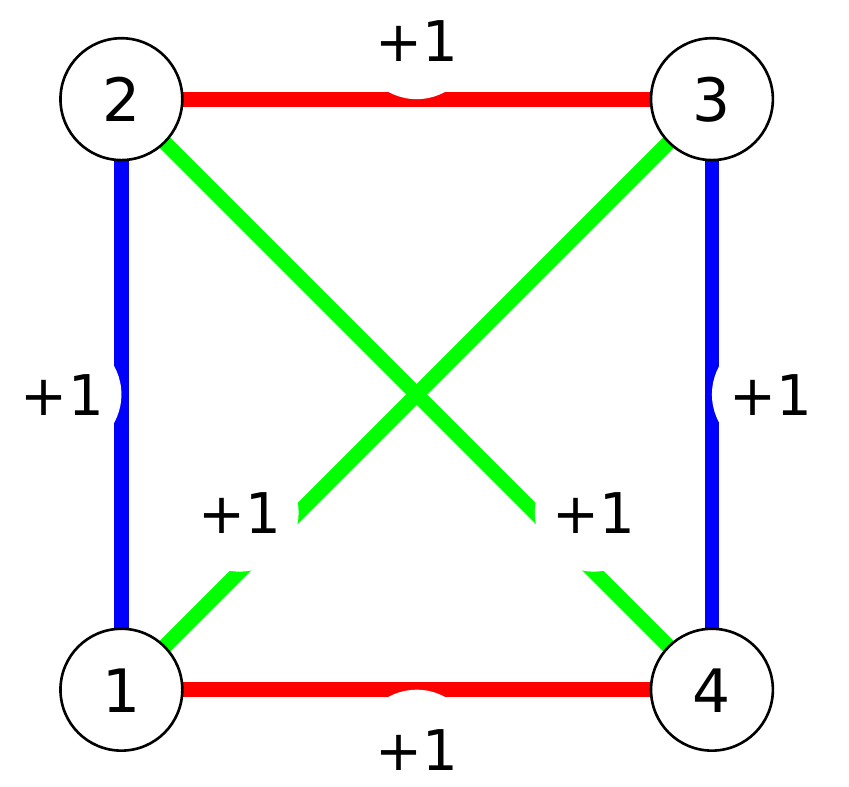}}
         \caption{$K_4$ with $3$ dimensions. $(1,2),(3,4)$ are of colour blue. $(1,3),(2,4)$ are of colour green. $(1,4),(2,3)$ are of colour red.}
         \label{fig:3dimk4}
\end{subfigure}
\hfill
\begin{subfigure}[b]{0.5\textwidth}
         \centering
         {\includegraphics[width=\linewidth]{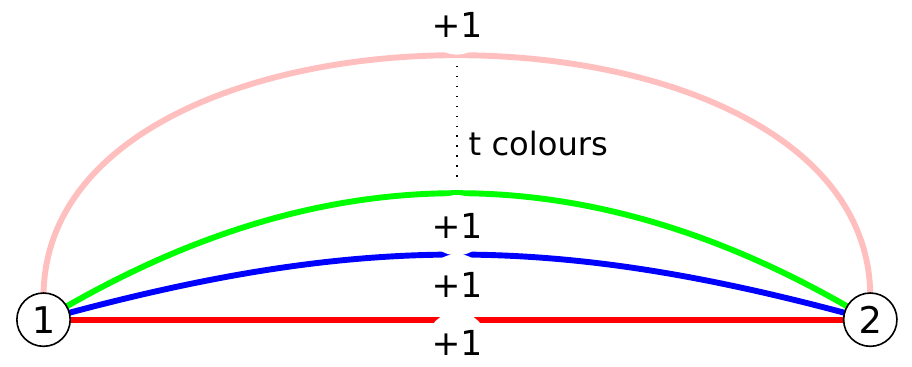}}
         \caption{A two-vertex multigraph with $t$ dimensions.}
         \label{fig:k2}
\end{subfigure}

\caption{GHZ graphs}
\label{fig:perfectly_mono_examples}
\end{figure}

A graph with a perfect matching has an even number of vertices. Most of our bounds concern graphs with an even number of vertices and at least four vertices; the edge-connectivity bound in \Cref{thm:edge_connectivity} also applies to graphs with two vertices. From \cref{fig:3dimk4}, we know that $\mu(K_4)\geq 3$ and, despite the use of huge computational resources \cite{AI1,AI2,neugebauer}, this is the only known simple graph on at least four vertices, up to isomorphism, with matching index at least $3$. Any graph with a matching index of at least $3$ and $n>4$ vertices would lead to a new GHZ state of dimension at least $3$ with $n>4$ entangled particles.
In the ordinary edge-colouring setting considered here, the Krenn--Gu conjecture asserts that 

\begin{conjecture}[{Krenn--Gu Conjecture}]
\label{krenn_conjecture}
For every loopless multigraph $G$ with $|V(G)|>4$, we have $\mu(G)\le 2$.
\end{conjecture}
{This conjecture has been extensively publicized, including through several cash rewards~\cite{mixon_website,krenn_website}. Its resolution would have direct implications for quantum optics. A counterexample would reveal a new destructive-interference mechanism in multi-photon systems and would immediately give a new construction of high-dimensional multipartite GHZ states. A proof, on the other hand, would identify a fundamental limitation of this optical model and would clarify which additional resources are necessary for producing larger high-dimensional GHZ states. Proving the conjecture for broad graph classes is also valuable, since it narrows the search space for computational tools such as PyTheus~\cite{AI2}.}

\subsection{Related work on the conjecture}
Some special cases of \cref{krenn_conjecture} have been studied by prior work. We discuss those in this subsection. 

\textbf{No destructive interference.} For the case of all edge weights being positive real numbers, it is easy to see that there cannot be any non-monochromatic perfect matching. Therefore, all perfect matchings are monochromatic. Bogdanov \cite{bogdanov,Quantum_graphs} observed that for this case, the matching index of $K_4$ is $3$, and for simple graphs non-isomorphic to $K_4$, the matching index is at most $2$. Due to this result, when there is no destructive interference, every simple matching-covered graph with a perfect matching, non-isomorphic to $K_4$, can achieve a maximum dimension of $1$ or $2$ and thus can be classified into Type $1$ and Type $2$ graphs (See \cite{CG1} for detailed discussion). Chandran and Gajjala~\cite{CG1} gave a structural classification for Type $2$ graphs. They further proved that for any half-edge colouring and edge weight assignment on a simple Type $2$ unweighted uncoloured graph, a dimension of $3$ or more cannot be achieved. The computational aspects of the vertex colourings arising from these experiments were studied by Vardi and Zhang \cite{DBLP:journals/corr/abs-2209-13063,VardiZ23}\rev{.}

\textbf{Absence of bichromatic edges or multi-edges.} 
In the full photonic model, the two photons produced by one source may have different 
mode numbers. Graph-theoretically, this is represented by giving the two ends of an 
edge different colours; such an edge is called \emph{bichromatic}, and the resulting 
colouring is sometimes called a \emph{half-edge colouring}. This half-edge model is 
more general than the ordinary edge-colouring model used in Definitions~1--4, where 
each edge has a single colour. The problem becomes easier when bichromatic edges or 
multi-edges are forbidden, and this restriction has led to several lines of work. 
Cervera-Lierta et al.~\cite{CerveraLierta2022designofquantum} used SAT solvers to 
prove that if the number of vertices is $6$ or $8$, the maximum dimension achievable 
is $2$ or, at most, $3$, respectively. They further conjectured that the maximum 
dimension achievable for an $n>4$ vertex graph is less than $\dfrac{n}{2}$ when 
bichromatic edges are absent. Chandran and Gajjala~\cite{CG2} proved that the 
maximum dimension achievable for an $n>4$ vertex simple GHZ graph, even when 
bichromatic edges are allowed, is less than $\dfrac{n}{\sqrt{2}}$.

\textbf{Unrestricted results.} {For the unrestricted setting, where both bichromatic edges and multi-edges are allowed, Chandran, Gajjala and Illickan~\cite{CG3} introduced a graph-reduction framework for the Krenn--Gu conjecture. They used this framework to prove the conjecture for graphs of vertex connectivity at most $2$ and for cubic graphs. Their argument relies on a computer-assisted result of Mantey~\cite{Kevin}, which shows that every graph on four vertices has matching index at most $3$. It remains open to give a short purely combinatorial proof of this four-vertex result.}

\subsection{Our Results}
{Our results give new structural evidence for the Krenn--Gu conjecture by showing that standard graph-connectivity parameters impose strong restrictions on the dimension of GHZ graphs. This is useful from both sides of the graph--physics bridge: it gives new graph-theoretic upper bounds on the matching index, and it rules out broad families of photonic experimental designs as candidates for high-dimensional GHZ-state generation. Such bounds are particularly relevant for computational searches using automated quantum-experiment design tools such as PyTheus~\cite{AI2}.}

Let $\Delta_{v}(G)$ denote the maximum number of neighbours of a vertex over all vertices, and let $\Delta_{e}(G)$ denote the maximum number of edges incident on a vertex over all vertices. Recall that, for multigraphs, $\Delta_{e}(G)\geq \Delta_{v}(G)$.
\begin{theorem}
\label{thm:max_degree}
For a multigraph $G$ with $|V(G)|>4$:
\begin{itemize}
    \item If $\Delta_{v}(G) \leq 3$, then $\mu(G) \leq 2 \leq |V(G)|-3$.
    \item If $\Delta_{v}(G) \geq 4$, then $\mu(G)\leq \Delta_{v}(G)-2 \leq |V(G)|-3$.
\end{itemize}
\end{theorem}
{In physical terms, Theorem~\ref{thm:max_degree} rules out the creation, within this graph model, of GHZ states of dimension at least $n-2$ using $n>4$ parties. A weaker bound, $\mu(G)<|V(G)|-1$, also follows from a result of~\cite{Kostochka2012LargeRM}. We prove Theorem~\ref{thm:max_degree} in \Cref{sec:max_degree}.} We also prove bounds in terms of edge connectivity.

{We next restrict attention to simple graphs, in which at most one edge joins any pair of vertices. For these graphs, we prove upper bounds in terms of vertex connectivity, and we obtain new combinatorial proofs for the smallest non-trivial cases.}

\begin{table}[t]
    \centering
    \caption{Known bounds, our results, and conjectures. The columns DI, ME, and BC indicate whether destructive interference, multi-edges, and bichromatic (half-edge) edges are allowed. Here $n=|V(G)|$ and $d$ is the dimension.}
    \label{tab:comparison1}
    \begin{tabularx}{\textwidth}{@{}ccc>{\raggedright\arraybackslash}p{0.38\textwidth}>{\raggedright\arraybackslash}X@{}}
        \hline
        \textbf{DI} & \textbf{ME} & \textbf{BC} & \textbf{Dimension bound} & \textbf{Status} \\
        \hline
        \xmark & \cmark & \cmark & $n\geq 4$, $d\leq 3$ & Proved in~\cite{Quantum_graphs}; characterized in~\cite{CG1}\\
        \cmark & \cmark & \cmark & $n=4$, $d\leq 3$ & Proved in~\cite{Kevin}\\
        \cmark & \cmark & \cmark & $n>4$, $d\leq 2$ & Proved for sparse graphs~\cite{CG3}\\
        \cmark & \xmark & \cmark & $n>4$, $d<n/\sqrt{2}$ & Proved in~\cite{CG2}\\
        \hline
        \cmark & \cmark & \xmark & $n>4$, $d\leq n-3$ & \Cref{thm:max_degree}\\
        \cmark & \cmark & \xmark & Connected $G$, $d\leq\lambda(G)$ & \Cref{thm:edge_connectivity}\\
        \cmark & \xmark & \xmark & Connected $G$, $n\geq 4$;\newline $d\leq 2^{\kappa(G)-1}+\binom{\kappa(G)}{2}$ & \Cref{thm:vertex_connectivity}\\
        \cmark & \cmark & \xmark & $n=6$, $d<3$; $n=8$, $d<4$ & Computer-assisted proof in~\cite{CerveraLierta2022designofquantum}\\
        \cmark & \xmark & \xmark & $n=6$, $d<3$; $n=8$, $d<4$ & \Cref{thm:six_eight_vertices}; combinatorial proof\\
        \hline
        \cmark & \cmark & \xmark & $n>4$, $d<n/2$ & Conjecture~\cite{CerveraLierta2022designofquantum}\\
        \cmark & \cmark & \cmark & $n>4$, $d\leq 2$ & Krenn--Gu conjecture~\cite{krenn_website}\\
        \hline
    \end{tabularx}
\end{table}

{For simple connected graphs, since one edge of each of the $\mu(G)$ feasible monochromatic colours must be incident on each vertex, $\mu(G)$ is at most the minimum degree.}
\begin{observation}
\label{obs:hat_min_degree}
{For a simple connected graph $G$, $\mu(G) \leq \delta(G)$.}
\end{observation}

{Observation~\ref{obs:hat_min_degree} gives the first degree-based upper bound for simple graphs. The parameters $\lambda(G)$ and $\kappa(G)$ are often more informative than $\delta(G)$, and Whitney's inequalities~\cite{whitney} give
\[
    \kappa(G) \leq \lambda(G) \leq \delta(G).
\]
It is therefore natural to ask whether $\mu(G)$ can be bounded directly in terms of edge or vertex connectivity. Our first connectivity bound answers this for edge connectivity; in fact, it holds for loopless multigraphs.}

\begin{theorem}
\label{thm:edge_connectivity}
{For every connected multigraph $G$, $\mu(G) \leq \lambda(G)$.}
\end{theorem}
{We present the proof of \cref{thm:edge_connectivity} in \cref{sec:edge_connectivity}.}

We also give an upper bound for $\mu(G)$ in terms of vertex connectivity, $\kappa(G)$, for simple graphs.
\begin{theorem}
\label{thm:vertex_connectivity}
{For every simple connected graph $G$ with $|V(G)|\geq 4$, 
\[
    \mu(G) \leq 2^{\kappa(G)-1}+ \binom{\kappa(G)}{2}.
\]}
\end{theorem}
We note that when $G$ is disconnected and has a perfect matching, $\mu(G)=1$. Also, when $G$ is disconnected and has no perfect matching, $\mu(G)=0$. We present the proof for \cref{thm:vertex_connectivity} in \cref{sec:vertex_connectivity}.

{Finally, we revisit the smallest non-trivial cases. Cervera-Lierta et al.~\cite{CerveraLierta2022designofquantum} used SAT solvers to show that graphs on six and eight vertices have maximum dimension at most $2$ and $3$, respectively, in a more general setting allowing multi-edges. We give a purely combinatorial proof for simple graphs.}
\begin{theorem}
\label{thm:six_eight_vertices}
{For a simple graph, if $|V(G)|= 6$ then $\mu(G) \leq 2$ and if $|V(G)|= 8$ then $\mu(G) \leq 3$.}
\end{theorem}
 We present the proof of \cref{thm:six_eight_vertices} in \cref{sec:six_vertices}. 

{Theorem~\ref{thm:edge_connectivity} is tight for $K_4$: its edge connectivity is $3$, and its matching index is also $3$. Theorem~\ref{thm:vertex_connectivity} is tight for $\kappa(G)=1$, for instance for paths with a perfect matching. For $\kappa(G)=2$, stronger results are already known~\cite{CG3}; for higher connectivity, the optimal dependence on $\kappa(G)$ remains open, and in particular it is not known whether a polynomial bound in $\kappa(G)$ should hold. Indeed, any example on more than four vertices with matching index greater than $2$ would be a breakthrough, both as a counterexample to the Krenn--Gu conjecture and as a new route to high-dimensional GHZ-state generation.}

\section{Preliminaries}
\label{sec:prelim}
Let $G_c^w$ be a GHZ graph. We call a colour class $i$ of the colouring $c$ \emph{redundant} if it does not contain any monochromatic perfect matching of $G_c^w$. Suppose we remove all the edges of such a redundant colour class $i$ to form a new graph ${G'}_c^w$. Observe that any feasible vertex colouring of the new graph ${G'}_c^w$ has the same weight in ${G}_c^w$ and ${G'}_c^w$. Therefore, an edge-coloured edge-weighted graph remains a GHZ graph, even after removing all its edges from a redundant colour class. Moreover, the number of monochromatic vertex colourings of weight $1$ remains the same. Therefore, such a removal does not affect the dimension of the edge-coloured edge-weighted graph. The reader may note that a redundant colour class may have non-redundant edges.

{For $S\subseteq V(G)$, let $G[S]$ denote the subgraph induced by $S$. We write $\mathbb{N}$ and $\mathbb{C}$ for the natural and complex numbers, respectively. The cardinality of a set $\mathcal{S}$ is denoted by $|\mathcal{S}|$. For a positive integer $r$, let $[r]=\{1,2,\ldots,r\}$.}

For $S\subseteq V(G)$, we define $\mathbf{i}_S$ to be the monochromatic inherited vertex colouring using colour $i$ over the induced subgraph $G_c^w[S]$. Its weight is denoted by $w(\mathbf{i}_S)$. We emphasize that this weight is with respect to the induced subgraph $G_c^w[S]$. {We now prove factorisation observations that are used throughout the proofs. These factorisations hold for graphs of any order.}

\begin{observation}
\label{colour_split_lemma}
{Let $vc$ be an inherited vertex colouring using the colours $[r]$, and let
$V_i=\{v\in V(G): vc(v)=i\}$
for each $i\in [r]$. Then
\[
    w(vc)=\prod_{i=1}^{r}w(\mathbf{i}_{V_i}).
\]}
\end{observation}
\begin{proof}
{A perfect matching inducing $vc$ is exactly the disjoint union, over the colours $i$ used by $vc$, of a monochromatic perfect matching of colour $i$ in $G_c^w[V_i]$. Conversely, any such collection gives a perfect matching inducing $vc$. Summing products of edge weights over these independent choices gives the claimed product.}
\end{proof}

Let $G_c^w[\mathbf{i},S]$ denote the subgraph of $G_c^w[S]$ formed by all edges of $G_c^w[S]$ from colour class $i$.
\begin{observation}
\label{partition_split_lemma}
{For $S \subseteq V(G)$, let the connected components of $G_c^w[\mathbf{i},S]$ have vertex sets $S_1,\ldots,S_r$. Then
\[
    w(\mathbf{i}_S)=\prod_{j=1}^{r}w(\mathbf{i}_{S_j}).
\]}
\end{observation}
\begin{proof}
{A colour-$i$ perfect matching of $G_c^w[S]$ is the disjoint union of colour-$i$ perfect matchings of the connected components of $G_c^w[\mathbf{i},S]$. The weight sum therefore factorises over these components.}
\end{proof}

Consider an edge $e=(v_1,v_2)$ of colour $c(e)$. A monochromatic inherited vertex colouring of colour $c(e)$ can be inherited from perfect matchings of two types. 
\begin{enumerate}
    \item The monochromatic perfect matchings of colour $c(e)$ containing $e$.
    \item The monochromatic perfect matchings of colour $c(e)$ not containing $e$. 
\end{enumerate}
{If the sum of weights of perfect matchings of type $1$ is non-zero, we call the edge $e$ a \emph{contributing} edge. If it is $0$, we call it a non-contributing edge.}

For each pair of vertices and each colour, replace all parallel edges of that colour by a single edge whose weight is the sum of their weights. If this sum is zero, delete the entire bundle instead. This operation preserves the weight of every inherited vertex colouring: the contribution of the parallel edges factors as the sum of their weights. In particular, the weights of all feasible monochromatic colourings remain $1$, so the graph remains a GHZ graph of the same dimension. The resulting subgraph has no more than one edge of each colour between any two vertices, and neither its neighbour-degrees nor its edge-degrees have increased. We use this reduction for the multigraph bounds below. The elementary bound $\mu(G)\leq\delta_e(G)$ follows because each feasible monochromatic colour has an incident edge at every vertex. We strengthen this bound using \cref{obs:contri_edges}.
\begin{observation}
\label{obs:contri_edges}
After the parallel-edge reduction, in a GHZ graph $G_c^w$ with $|V(G)|\geq 4$, if there are two edges $e$ and $e'$ of different colours between $u,v$, then both must be non-contributing.  
\end{observation}
\begin{proof}
Towards a contradiction, let there be a contributing edge $e=(v_1,v_2)$ and another edge $e'=(v_1,v_2)$ such that $c(e')\neq c(e)$. Let $V(G)=V$ and $V'=V\setminus\{v_1,v_2\}$. Let $c(e)={i}$. Since $e$ is contributing, we know that $w(e)w(\mathbf{i}_{V'}) \neq 0$ by definition. Therefore, $w(\mathbf{i}_{V'}) \neq 0$. Let $c(e')={j}$. Consider the non-monochromatic colouring in which $\{v_1,v_2\}$ get colour $j$ and $V'$ get colour $i$. Since $e'$ is the unique colour-$j$ edge between these vertices, \cref{colour_split_lemma} gives weight $w(e')w(\mathbf{i}_{V'})\neq 0$. But a non-monochromatic colouring should have weight $0$. Contradiction. 
\end{proof}

\begin{observation}
\label{min_degree_bound}
{For a loopless multigraph $G$ with $n=|V(G)|\geq 4$, we have $\mu(G) \leq \delta_{v}(G) \leq n-1$.}
\end{observation}
\begin{proof}
{The claim is immediate when $\mu(G)=0$. Otherwise, apply the parallel-edge reduction to a colouring and weight assignment witnessing $\mu(G)$, and let the resulting GHZ graph be $H_c^w$. It has dimension $\mu(G)$ and $\delta_v(H)\leq\delta_v(G)$. Let $v$ be a vertex with exactly $\delta_v(H)$ neighbours. As each of the $\mu(G)$ monochromatic inherited vertex colourings has weight one, there must be at least one contributing edge of each of those $\mu(G)$ colours incident on $v$. If $\mu(G)>\delta_v(H)$, the pigeonhole principle gives a neighbour $u$ of $v$ with two contributing edges of different colours between $u$ and $v$, contradicting \cref{obs:contri_edges}. Therefore, $\mu(G)\leq\delta_v(H)\leq\delta_v(G)$.}
\end{proof}

\section{Maximum neighbour-degree: Proof of Theorem \ref{thm:max_degree}}
\label{sec:max_degree}
\begin{proof}
Let $d_{v}(u)$ denote the number of neighbours of a vertex $u \in G$. In $G_c^w$, let $d_{\mathbf{i}}(u)$ denote the number of edges of colour $i$ incident on $u$. Let $\Delta_{\mathbf{i}}(G_c^w) = \max\limits_{v \in V(G)}  d_{\mathbf{i}}(v)$ and $\delta_{\mathbf{i}}(G_c^w) = \min\limits_{v \in V(G)}  d_{\mathbf{i}}(v)$. We prove the following three lemmas first. 
\begin{lemma}
\label{max_deg_two}
In a GHZ graph $G_c^w$ with $|V(G)|>4$, if $\Delta_{\mathbf{i}}(G_c^w) \leq 2$, for all colours $i$, then $\mu(G,c,w)\leq 2$.
\end{lemma}
\begin{proof}
Note that if there is a counterexample $G_c^w$, by removing all of its redundant colours, we can construct a counterexample with no redundant colours. Therefore, it is sufficient to prove the statement for edge-coloured edge-weighted graphs with no redundant colours. Consider such a GHZ graph $G_c^w$.

If $G_c^w$ has no non-monochromatic perfect matchings and $|V(G)|>4$, then by the result of Bogdanov, $\mu(G,c,w)\leq 2$ \cite{Quantum_graphs}.

Therefore, let us assume that there is a non-monochromatic matching and its associated non-monochromatic inherited vertex colouring $vc$ on colours $[r]$ has weight zero. From \cref{colour_split_lemma}, there is a non-redundant colour in $[r]$, say $1$, for which $w(\mathbf{1}_{S})$ is zero, where 
$S=\{v \in V(G) \mid vc(v)=1\}$. From \cref{partition_split_lemma}, there is also a connected component $G_c^w[\mathbf{1},P]$ of $G_c^w[\mathbf{1},S]$ such that its weight, $w(\mathbf{1}_{P})=0$.

Since, $\Delta_{\mathbf{1}}(G_c^w) \leq 2$, $G_c^w[\mathbf{1},P]$ must be a cycle or path. Also, since $G_c^w[\mathbf{1},P]$ has a perfect matching in it, it must be an even cycle or a path with an odd number of edges. If it is a path with an odd number of edges, say $e_1e_2\ldots$, then  $w(\mathbf{1}_{P})=w(e_1)w(e_3)\ldots$, which cannot be zero. Therefore, it must be an even cycle.

Now consider the monochromatic inherited vertex colouring $\mathbf{1}_{V(G)}$. Note that $G_c^w[\mathbf{1},P]$ is disconnected from $G_c^w[\mathbf{1},{V(G)-P}]$ as it is a cycle and $\Delta_{\mathbf{1}}(G_c^w) \leq 2$. Therefore, from \cref{partition_split_lemma}, $w(\mathbf{1}_{V(G)})$ must be $0$, since $w(\mathbf{1}_{P})$ is zero. But $1$ is a non-redundant colour and hence its monochromatic inherited vertex colouring has weight $1$. Contradiction.
\end{proof}

\begin{lemma}
For a graph with $|V(G)|>4$, if $\Delta_{v}(G) \leq 3$, then $\mu(G) \leq 2$.
\label{maximum_degree_3}
\end{lemma}

\begin{proof}
Towards a contradiction, suppose that $\mu(G)\ge 3$.  Then there exists a
colouring $c$ and a weight function $w$ such that $\mu(G,c,w)\ge 3$.

Apply the parallel-edge reduction and delete all redundant colour classes.
These preserve the dimension and do not increase neighbour-degrees, so the
maximum remains at most $3$.  Relabelling the resulting graph as $G$, we may assume that every
colour class present is non-redundant.

If $\Delta_i(G^w_c)\le 2$ for every colour $i$, then \Cref{max_deg_two} gives
$\mu(G,c,w)\le 2$, a contradiction.  Hence there is a non-redundant colour,
say colour $1$, and a vertex $u$ such that $d_1(u)\ge 3$.

Since $\Delta_v(G)\le 3$, the vertex $u$ has exactly three neighbours, and each
of them is joined to $u$ by an edge of colour $1$.  Since colour $1$ is
non-redundant and $\mu(G,c,w)\ge 3$, there are at least two feasible colours
different from colour $1$.  For each such colour, there must be a contributing
edge of that colour incident with $u$.  Such an edge must go to one of the
three neighbours of $u$, but that neighbour is already joined to $u$ by an edge
of colour $1$.  By \Cref{obs:contri_edges}, the contributing edge would then be
non-contributing, a contradiction.
\end{proof}

\begin{lemma}
\label{max_deg_atleast_four}
{If $|V(G)|>4$ and $\Delta_v(G)\geq 4$, then
\[
    \mu(G)\leq \Delta_v(G)-2.
\]}
\end{lemma}

\begin{proof}
Let $D=\Delta_v(G)$.  Towards a contradiction, suppose that
$\mu(G)\ge D-1$.  Then there exists a colouring $c$ and a weight function $w$
such that $\mu(G,c,w)\ge D-1$.  Since $D\ge 4$, this implies
$\mu(G,c,w)\ge 3$.

Apply the parallel-edge reduction and delete all redundant colour classes.
These preserve the dimension and do not increase neighbour-degrees, so the
maximum remains at most $D$.  Relabelling the resulting graph as $G$, we may assume that every
colour class present is non-redundant.

If $\Delta_i(G^w_c)\le 2$ for every colour $i$, then \Cref{max_deg_two} gives
$\mu(G,c,w)\le 2$, a contradiction.  Hence there is a non-redundant colour,
say colour $1$, and a vertex $u$ such that $d_1(u)\ge 3$.

Since colour $1$ is feasible and $\mu(G,c,w)\ge D-1$, there are at least
$D-2$ feasible colours different from colour $1$.  For each such colour, choose
one contributing edge of that colour incident with $u$.  By \Cref{obs:contri_edges},
none of these contributing edges can go to a neighbour of $u$ that is already
joined to $u$ by an edge of colour $1$.  Also by \Cref{obs:contri_edges}, two such
contributing edges of different colours cannot go to the same neighbour.

Therefore $u$ has at least
\[
d_1(u)+(D-2) \ge 3+(D-2)=D+1
\]
distinct neighbours, contradicting their upper bound $D$.
Thus the original graph has matching index at most $D-2$.
\end{proof}

{Since $\Delta_v(G) \leq |V(G)|-1$, Lemmas~\ref{maximum_degree_3} and~\ref{max_deg_atleast_four} imply Theorem~\ref{thm:max_degree}.}
\end{proof}

\section{Edge connectivity: Proof of Theorem \ref{thm:edge_connectivity}}
\label{sec:edge_connectivity}
\begin{proof}
Towards a contradiction, suppose that there is a loopless connected multigraph $G''$
such that
\[
    \mu(G'')\ge \lambda(G'')+1.
\]
Choose a colouring $c$ and a weight function $w$ witnessing this dimension.
Delete all redundant edges and then all redundant colour classes. The dimension
is unchanged. Let the resulting graph be $G$. Then
\[
    \mu(G,c,w)=\mu(G'')\ge \lambda(G'')+1.
\]
Moreover, $G$ is a subgraph of $G''$. Suppose that $G$ were disconnected,
with component vertex sets $C_1,\ldots,C_s$, where $s\geq 2$.
For each feasible monochromatic colour $i$, \Cref{partition_split_lemma},
with its factors grouped by components of $G$, gives
$1=w(\mathbf{i}_{V(G)})=\prod_{t=1}^{s}w(\mathbf{i}_{C_t})$.
If two such colours $i,j$ existed, assigning colour $i$ to $C_1$ and colour
$j$ to all other components would, by \Cref{colour_split_lemma} and the
same component factorisation, give the non-zero weight
$w(\mathbf{i}_{C_1})\prod_{t=2}^{s}w(\mathbf{j}_{C_t})$.
This contradicts the GHZ condition. Thus a disconnected $G$ would satisfy
$\mu(G,c,w)\leq 1$, contradicting $\lambda(G'')\geq 1$ and
$\mu(G,c,w)\geq\lambda(G'')+1$. Hence $G$ is connected.
In particular,
\[
    \lambda(G)\le \lambda(G'').
\]

Let $(S,S')$ be a minimum edge cut of $G$, so $S'=V(G)\setminus S$ and
$|(S,S')|=\lambda(G)$. Since $G$ has a perfect matching, $|V(G)|$ is even,
and therefore $|S|$ and $|S'|$ have the same parity.

First suppose that $|S|$ and $|S'|$ are odd.  Every monochromatic perfect
matching must contain at least one edge of the cut $(S,S')$.  Therefore every
feasible monochromatic colour appears on at least one edge of the cut. Hence
\[
\mu(G,c,w)\le |(S,S')|=\lambda(G)\le \lambda(G''),
\]
contradicting $\mu(G,c,w)\ge \lambda(G'')+1$.

Now suppose that $|S|$ and $|S'|$ are even.  Let $A$ be the set of feasible
colours for which no monochromatic perfect matching crosses the cut
$(S,S')$, and put $r=|A|$.  Let $B$ be the remaining feasible colours and put
$r'=|B|$.  Thus
\[
\mu(G,c,w)=r+r'.
\]

For every colour $i\in B$, some colour-$i$ monochromatic perfect matching
crosses the cut.  Since both sides of the cut have even cardinality, such a
matching uses at least two cut edges of colour $i$.  Therefore
\[
\lambda(G)=|(S,S')|\ge 2r'.
\]

We claim that $r\le 1$.  Suppose not, and choose two colours $1,2\in A$.
Since no colour-$i$ monochromatic perfect matching crosses the cut for
$i=1,2$, all such matchings split between the two sides of the cut.  Hence
\[
w(1_{V(G)})=w(1_S)w(1_{S'})=1
\quad\text{and}\quad
w(2_{V(G)})=w(2_S)w(2_{S'})=1.
\]
In particular, $w(1_S)$ and $w(2_{S'})$ are non-zero.

Consider the inherited vertex colouring $vc$ in which every vertex of $S$ gets
colour $1$ and every vertex of $S'$ gets colour $2$.  Any perfect matching
inducing $vc$ contains no cut edge: a cut edge in an ordinary edge-colouring
would give the same colour to its two endpoints.  Therefore, by
\Cref{colour_split_lemma},
\[
w(vc)=w(1_S)w(2_{S'})\ne 0.
\]
This is impossible because $vc$ is non-monochromatic.  Hence $r\le 1$.

If $r'\ge 1$, then
\[
\mu(G,c,w)=r+r'\le 1+r'\le 2r'\le \lambda(G)\le \lambda(G''),
\]
again contradicting $\mu(G,c,w)\ge \lambda(G'')+1$.  If $r'=0$, then
$\mu(G,c,w)=r\le 1$, also a contradiction.  Therefore no such $G''$ exists.
\end{proof}

\section{Vertex connectivity for simple graphs: Proof of Theorem \ref{thm:vertex_connectivity}}
\label{sec:vertex_connectivity}
\begin{proof}
If $G$ is a complete graph, from \cref{min_degree_bound}, $\mu(G) \leq |V(G)|-1$. As $\kappa(G)=|V(G)|-1$, it is easy to see that the given statement holds. We now consider simple connected graphs not isomorphic to complete graphs.
\begin{lemma}
\label{vertex_conn_strong_version}
If $T$ is a minimum vertex separator of $G$, then
\[
    \mu(G) \leq 2^{|T|-1}+|E(T)|.
\]
\end{lemma}
\begin{proof}
Let $k=|T|=\kappa(G)$. Let a connected component of $G[V(G)-T]$ be $G[S]$, and let $S'=V(G)-T-S$. Towards a contradiction, assume that
\[
    \mu(G) \geq 2^{k-1}+|E(T)|+1.
\]
Therefore, there exists a colouring $c$ and a weight function $w$ such that
\[
    \mu(G)=\mu(G,c,w) \geq 2^{k-1}+|E(T)|+1.
\]
Since there are $|E(T)|$ edges with both endpoints in $T$, there can be at most $|E(T)|$ colour classes containing an edge in $T$. Therefore, there are at least $2^{k-1}+1$ non-redundant colour classes that do not contain any edge from $G_c^w[T]$. Let those non-redundant colour classes be $\{1,2,\ldots,2^{k-1}+1\}$.

For each $i\in [2^{k-1}+1]$, let $M^i_{V(G)}\ne\emptyset$ be the set of
colour-$i$ monochromatic perfect matchings of $G^w_c$.

For $J\subseteq T$, write $J^c=T\setminus J$, and define
\[
a_i(J)=w(i_{S\cup J})
\quad\text{and}\quad
b_i(J)=w(i_{S'\cup J^c}).
\]
Let $
\mathcal{T}=\{J\subseteq T:\ |S\cup J|\text{ is even}\}.$
Then $|\mathcal{T}|=2^{k-1}$.

Since the colours under consideration have no edge inside $T$, every
colour-$i$ monochromatic perfect matching of $G$ splits uniquely as a
colour-$i$ perfect matching of $G[S\cup J]$ together with a colour-$i$ perfect
matching of $G[S'\cup J^c]$, for a unique $J\in\mathcal{T}$.  Therefore
\[
\sum_{J\in\mathcal{T}} a_i(J)b_i(J)
=
w(i_{V(G)})
=
1.
\]
Hence, for each such colour $i$, there exists $J_i\in\mathcal{T}$ such that
$
a_i(J_i)b_i(J_i)\ne 0.
$
Since there are $2^{k-1}+1$ colours and only
$2^{k-1}$ choices for $J_i$, two colours, say $1$ and $2$, satisfy
$J_1=J_2=J$.

Thus
\[
a_1(J)b_1(J)\ne 0
\quad\text{and}\quad
a_2(J)b_2(J)\ne 0,
\]
so in particular
\[
a_1(J)b_2(J)\ne 0.
\]
Consider the inherited vertex colouring in which $S\cup J$ is coloured $1$ and
$S'\cup J^c$ is coloured $2$.  By \Cref{colour_split_lemma}, its weight is
\[
a_1(J)b_2(J)\ne 0,
\]
contradicting the GHZ condition because this inherited vertex colouring is
non-monochromatic.
\end{proof}
Since $T$ is a minimum vertex separator, $|T|=\kappa(G)$, and since $G$ is simple we have $|E(T)|\leq \binom{\kappa(G)}{2}$. Hence \cref{thm:vertex_connectivity} follows.
\end{proof}

\section{Graphs with six or eight vertices: Proof of Theorem \ref{thm:six_eight_vertices}}
\label{sec:six_vertices}
\subsection{Six-vertex case}
\begin{proof}
Towards a contradiction, let $\mu(G) \geq 3$. Therefore, there exists a  colouring $c$ and weight function $w$, such that $\mu(G,c,w)\geq 3$. Let three of its non-redundant colours be $1$, $2$ and $3$. Consider the union of a colour $1$ perfect matching and a colour $2$ perfect matching. Their union is a disjoint union of even cycles whose edges alternate between colours $1,2$. Since $G$ is simple and the two matchings are edge-disjoint, every such cycle has length at least $4$. With only $6$ vertices, their union must therefore be a single $6$-cycle. Let the vertices of such a $6$ cycle be labelled $abcdef$ in clockwise direction. Without loss of generality, let $c(ab)=c(cd)=c(ef)=1$ and $c(bc)=c(de)=c(fa)=2$.

{We call $(a,d)$, $(b,e)$, and $(c,f)$ the antipodal pairs.} We claim that any perfect matching $M$ formed by non-cyclic edges must have at least one edge containing an antipodal pair. Suppose not. Observe that any edge connecting a vertex from $\{a,c,e\}$ to $\{b,d,f\}$ is either a cyclic edge or contains an antipodal pair. {Thus $M$ would have to match the vertices of $\{a,c,e\}$ among themselves and the vertices of $\{b,d,f\}$ among themselves, which is impossible because both sets have odd cardinality.}

Therefore, for the monochromatic matching of colour $3$, there is an edge containing an antipodal pair, say $ad$, without loss of generality. Consider the non-monochromatic inherited vertex colouring in which $a,d$ are coloured $3$, $b,c$ are coloured $2$ and $e,f$ are coloured $1$. By \cref{colour_split_lemma}, the weight of such a vertex colouring is $w(ad)w(bc)w(ef)$. Since the vertex colouring is non-monochromatic, $w(vc) = 0$. Therefore, one of these three edges must have weight zero. Contradiction.
\end{proof} 

\subsection{Eight-vertex case for Theorem~\ref{thm:six_eight_vertices}}
\label{sec:eight_vertices}
For the eight vertex case, the proof has two main cases. In Case~1, every pair of monochromatic perfect matchings forms two disjoint $4$-cycles; we show that a third colour then forces an alternating Hamiltonian cycle, contradicting the case assumption. In Case~2, we fix an alternating red-green Hamiltonian cycle and classify the remaining colours according to whether one of their matching edges separates the cycle into two even parts. Figures~\ref{fig:eightv_c3}--\ref{fig:eightv_c4} list the resulting configurations up to dihedral symmetry and interchange of colour names.

\begin{proof}
Towards a contradiction, let $\mu(G) \geq 4$. Therefore, there exists a  colouring $c$ and weight function $w$, such that $\mu(G,c,w)\geq 4$. Let four of its non-redundant colours be $1$ (red), $2$ (green), $3$ (blue) and $4$ (brown). Consider the union of a perfect matching of two different colours. Their union is a disjoint union of even cycles with alternating colours. Since $G$ is simple and the matchings are edge-disjoint, every such cycle has length at least $4$. With $8$ vertices, their union is therefore either an $8$-cycle or two $4$-cycles.

\textbf{Case 1.} We first consider the case where the union of any two
monochromatic perfect matchings of different colours is always two disjoint
$4$-cycles.  Consider a red and a green perfect matching.  Their union forms
two disjoint cycles $1234$ and $5678$, where
\[
(1,2),(3,4),(5,6),(7,8)
\]
are red and
\[
(1,4),(2,3),(5,8),(6,7)
\]
are green.

Consider the non-monochromatic vertex colouring in which $1234$ gets colour
red and $5678$ gets colour green.  By \Cref{colour_split_lemma}, this colouring has weight
\[
w(1_{\{1,2,3,4\}})\,w(2_{\{5,6,7,8\}}),
\]
and this weight must be zero.  Hence one of the two factors is zero.  Without
loss of generality, assume
\[
w(1_{\{1,2,3,4\}})=0.
\]
Since the red edges $(1,2)$ and $(3,4)$ are present, this forces the red edges
$(1,3)$ and $(2,4)$ to be present and
\[
w((1,2))w((3,4))+w((1,3))w((2,4))=0.
\]

Now consider the vertex colouring in which $1234$ gets colour green and
$5678$ gets colour red.  Again by \Cref{colour_split_lemma}, one of
$w(2_{\{1,2,3,4\}})$ and $w(1_{\{5,6,7,8\}})$ must be zero.  But
\[
w(2_{\{1,2,3,4\}})=w((1,4))w((2,3))\ne 0.
\]
Therefore
\[
w(1_{\{5,6,7,8\}})=0.
\]
This forces the red edges $(5,7)$ and $(6,8)$ to be present and
\[
w((5,6))w((7,8))+w((5,7))w((6,8))=0.
\] We know that there should be a blue perfect matching. Note that {all possible edges within $\{1,2,3,4\}$ and within $\{5,6,7,8\}$ are red or green.} Therefore, all blue edges must have one vertex from $1,2,3,4$ and the other from $5,6,7,8$. Without loss of generality, let $(2,5)$ be an edge from the blue perfect matching. 
\begin{figure}[t!]
    \centering
\centering    
\begin{subfigure}[b]{0.4\textwidth}
         \centering
         \fbox{\includegraphics[width=50mm]{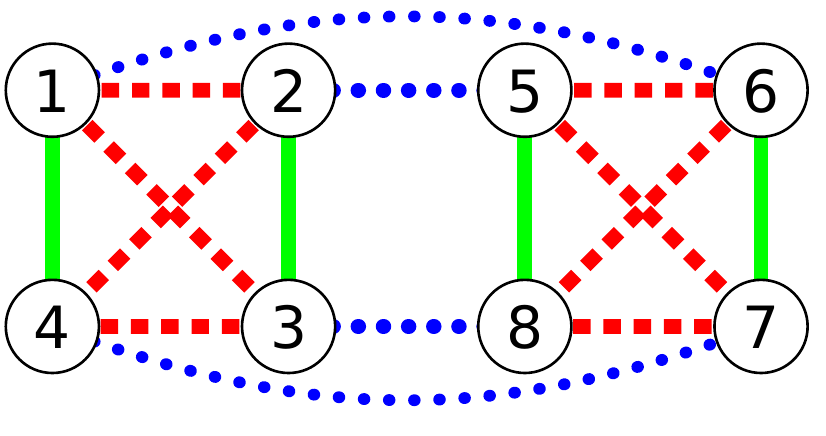}}
         \caption{Case 1.1.1}
         \label{fig:c31}
\end{subfigure}
\hfill
\begin{subfigure}[b]{0.4\textwidth}
         \centering
         \fbox{\includegraphics[width=50mm]{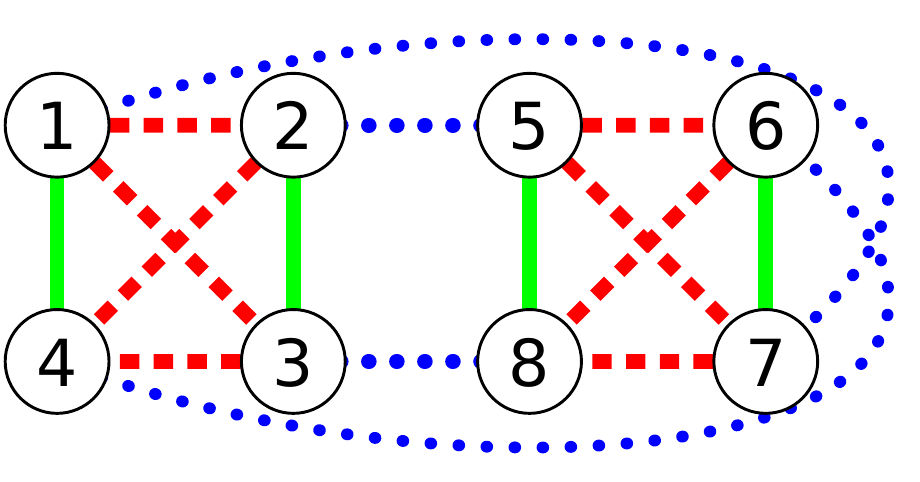}}
         \caption{Case 1.1.2 }
         \label{fig:c32}
\end{subfigure}
\hfill
\begin{subfigure}[b]{0.4\textwidth}
         \centering
         \fbox{\includegraphics[width=50mm]{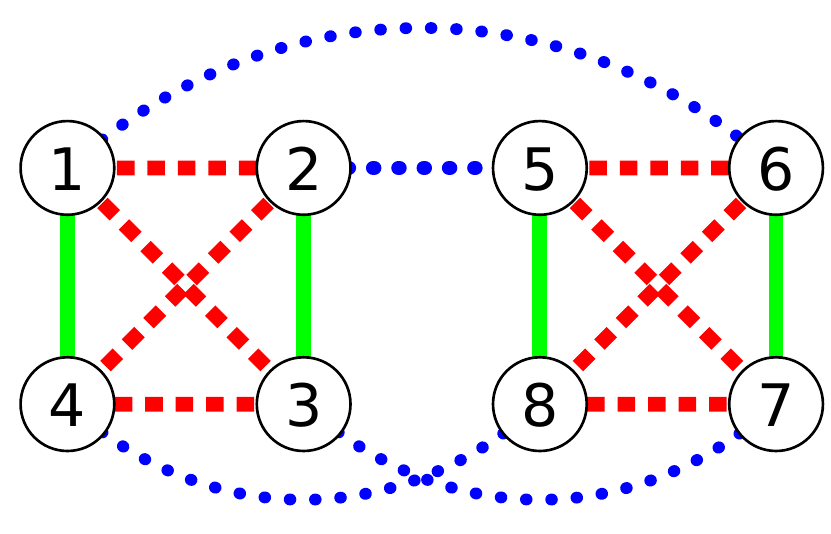}}
         \caption{Case 1.2.1 }
         \label{fig:c33}
\end{subfigure}
\hfill
\begin{subfigure}[b]{0.4\textwidth}
         \centering
         \fbox{\includegraphics[width=50mm]{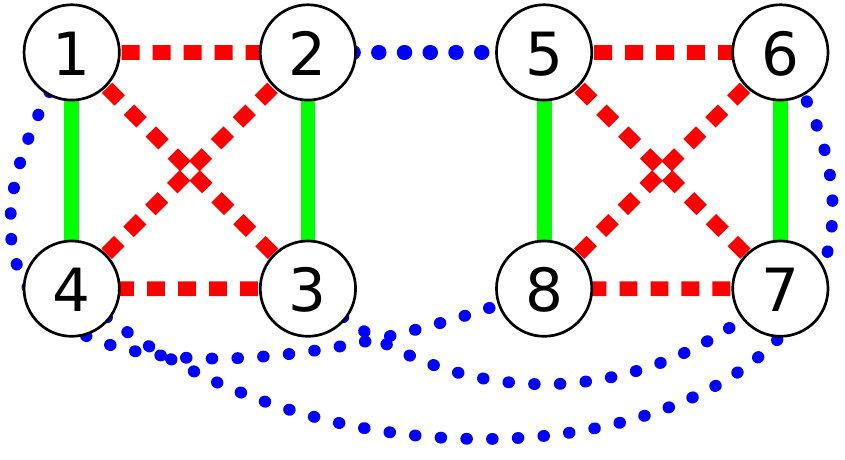}}
         \caption{Case 1.2.2}
         \label{fig:c34}
\end{subfigure}

\caption{Case 1. Edge styles in Figures~\ref{fig:eightv_c3}--\ref{fig:eightv_c4}: red dashed, green solid, blue dotted, and brown dash-dotted.}
\label{fig:eightv_c3}
\end{figure}

\begin{itemize}
    \item \textbf{Case 1.1} The blue perfect matching contains $(3,8)$. Then the blue perfect matching can be completed in two ways as shown in \cref{fig:c31} and \cref{fig:c32}. 
    \begin{itemize}
        \item \textbf{Case 1.1.1} (\cref{fig:c31}.) There is a blue-red alternating Hamiltonian cycle $16834752$ which is a union of two perfect matchings of different colours. By assumption such perfect matchings are not present in this case. Contradiction.
        \item \textbf{Case 1.1.2} (\cref{fig:c32}.) There is a blue-red alternating Hamiltonian cycle $17834652$ which is a union of two perfect matchings of different colours. By assumption such perfect matchings are not present in this case. Contradiction.
    \end{itemize}
    \item \textbf{Case 1.2} Vertex $3$ can have an edge with $6$ or $7$. Note that both are identical red neighbours and hence without loss of generality let the edge be with $7$. Then the blue perfect matching can be completed in two ways. 
\begin{itemize}
\item \textbf{Case 1.2.1} (\cref{fig:c33}.) There is a blue-red alternating Hamiltonian cycle $16843752$. It is a union of two perfect matchings of different colours. Contradiction.
\item \textbf{Case 1.2.2} (\cref{fig:c34}.) There is a blue-red alternating Hamiltonian cycle $12564378$. It is a union of two perfect matchings of different colours. Contradiction.
    \end{itemize}
\end{itemize}

\textbf{Case 2.} We classify the remaining cases with respect to the fixed alternating
red-green Hamiltonian cycle. Label its vertices $1,\ldots,8$ in cyclic order.
For a monochromatic perfect matching $M$ of a colour
$i\notin\{\text{red},\text{green}\}$, call an edge of $M$ \emph{crossing}
if removing its endpoints from the Hamiltonian cycle leaves two paths of
even order. Say that $i$ is crossing if some such $M$ contains a crossing
edge. There are two cases.
\begin{enumerate}
    \item  Some remaining non-redundant colour, say blue, is not crossing.  Then every edge of
every blue monochromatic perfect matching joins vertices of the same parity in
the cyclic ordering.
\item Every non-redundant colour other than red and green is crossing.
\end{enumerate}
\begin{figure}[t!]
    \centering
\centering    
\begin{subfigure}[b]{0.3\textwidth}
         \centering
         \fbox{\includegraphics[width=\dimexpr\linewidth-2\fboxsep-2\fboxrule\relax]{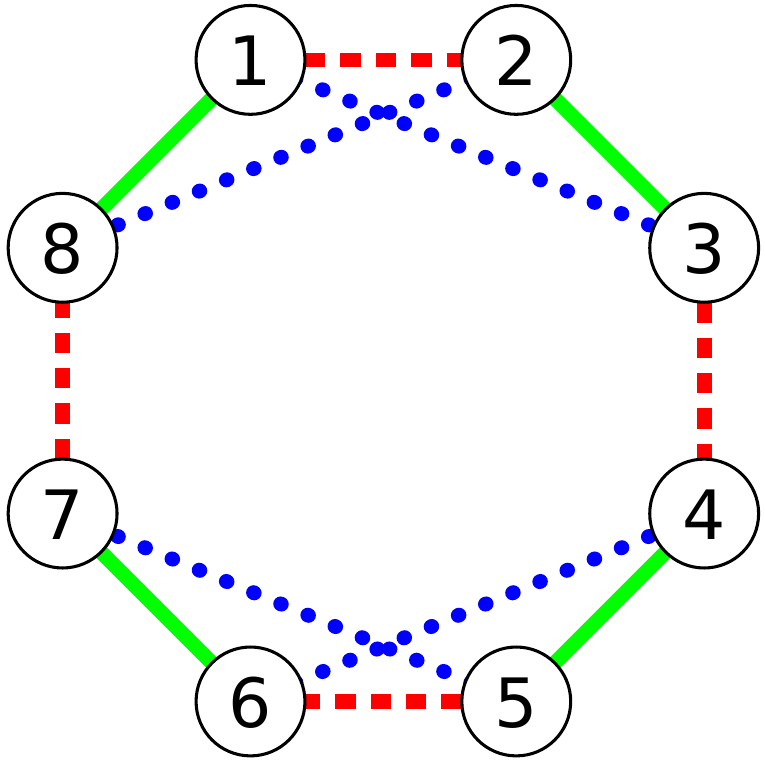}}
         \caption{Case 2.1.1}
         \label{fig:c11}
\end{subfigure}
\hfill
\begin{subfigure}[b]{0.3\textwidth}
         \centering
         \fbox{\includegraphics[width=\dimexpr\linewidth-2\fboxsep-2\fboxrule\relax]{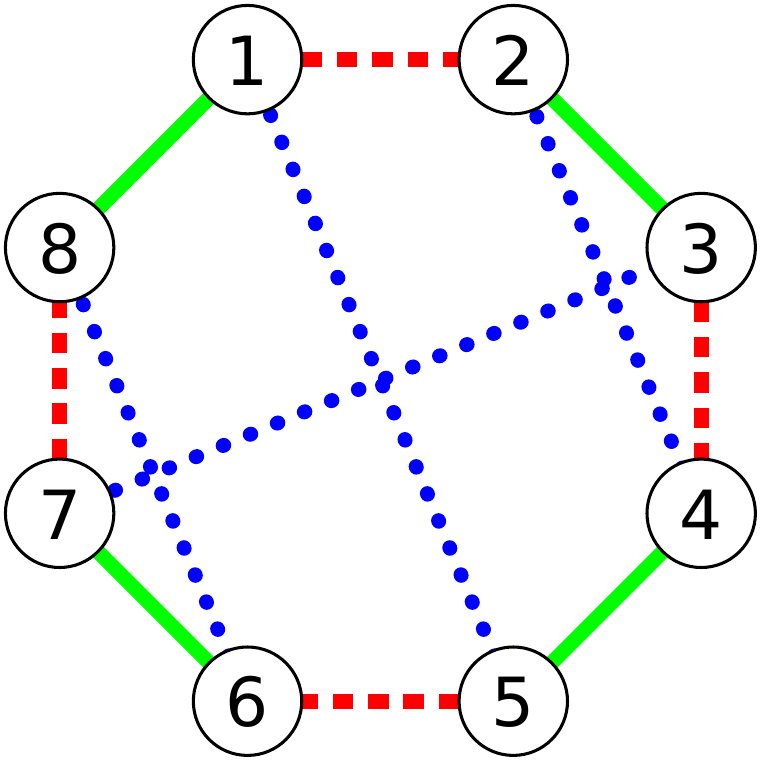}}
         \caption{Case 2.1.2 }
         \label{fig:c12}
\end{subfigure}
\hfill
\begin{subfigure}[b]{0.3\textwidth}
         \centering
         \fbox{\includegraphics[width=\dimexpr\linewidth-2\fboxsep-2\fboxrule\relax]{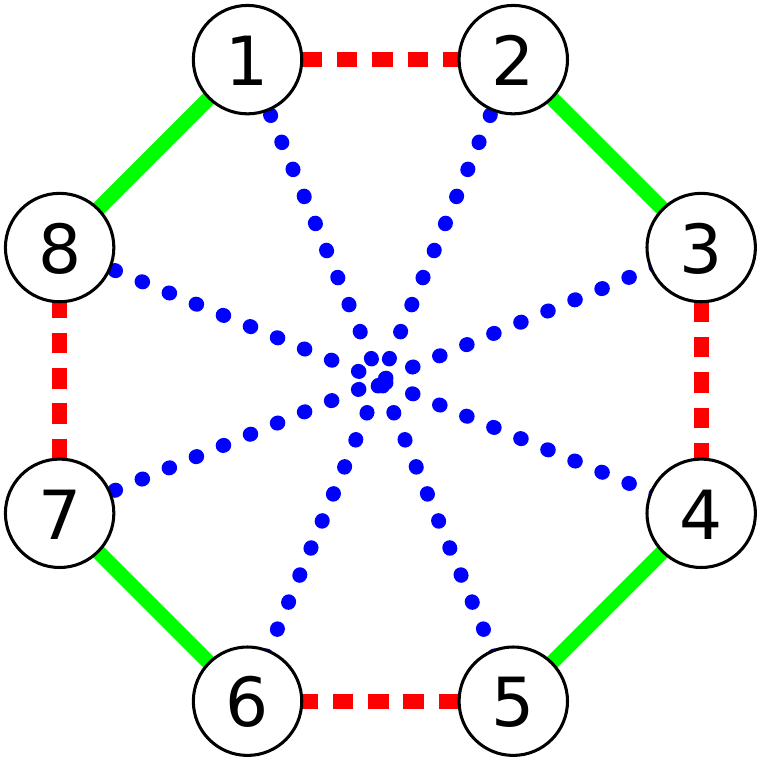}}
         \caption{Case 2.1.3}
         \label{fig:c13}
\end{subfigure}

\caption{Case 2.1}
\label{fig:eightv_c1}
\end{figure}

\textbf{Case 2.1.} Fix a blue monochromatic perfect matching $M_b$.
Since blue is not crossing, $M_b$ pairs vertices of equal parity: the four
odd vertices among themselves and the four even vertices among themselves.
Up to dihedral symmetry of the Hamiltonian cycle and interchange of red
and green, its three possible types are:
\[
\begin{aligned}
B_1&=\{(1,3),(2,8),(4,6),(5,7)\},\\
B_2&=\{(1,5),(2,4),(3,7),(6,8)\},\\
B_3&=\{(1,5),(2,6),(3,7),(4,8)\}.
\end{aligned}
\]
These are precisely the three cases shown in \Cref{fig:eightv_c1}.

\textbf{Case 2.1.1} (\cref{fig:c11}).  Consider the non-monochromatic vertex
colouring in which $1,2,3,8$ get colour blue and $4,5,6,7$ get colour green.
By \Cref{colour_split_lemma}, the weight of this colouring is
\[
w((1,3))w((2,8))w((4,5))w((6,7))\ne 0.
\]
This contradicts the GHZ condition, since every non-monochromatic inherited
vertex colouring must have weight $0$.

\textbf{Case 2.1.2} (\cref{fig:c12}). Consider the non-monochromatic vertex colouring in which $1,2$ get the colour red, $4,5$ get colour green and $3,6,7,8$ get the colour blue. From \cref{colour_split_lemma}, the weight of such a colouring is $w((1,2))w((4,5))w((3,7))w((6,8)) \neq 0$. But a non-monochromatic vertex colouring should have weight $0$. Contradiction.

We now record the structure forced by a crossing edge in a selected
monochromatic perfect matching $M_i$, for a colour $i$ other than red and
green. Up to a dihedral symmetry and interchange of red and green, let
$(8,3)$ be a crossing edge of $M_i$.
Consider the non-monochromatic vertex colouring assigning colour $i$ to
$3,8$, red to $1,2$, and green to $4,5,6,7$.
By \Cref{colour_split_lemma}, its weight is
$w((8,3))w((1,2))w(\mathbf{2}_{\{4,5,6,7\}})=0$.
The first two factors are non-zero, so
$w(\mathbf{2}_{\{4,5,6,7\}})=0$.
The green matching $(4,5),(6,7)$ has non-zero weight; its only possible
alternative is $(4,6),(5,7)$, since $(5,6)$ is red.
Thus $(4,6)$ and $(5,7)$ must both be green.
After $(8,3)$ is removed from $M_i$, four unmatched vertices lie in
$\{4,5,6,7\}$ and only two lie outside it. Therefore $M_i$ contains an
edge within $\{4,5,6,7\}$. The only available pair is $(4,7)$:
all other pairs already have red or green edges, and $G$ is simple.
The remaining vertices $1,2,5,6$ can be matched in exactly two ways,
$(1,5),(2,6)$ or $(1,6),(2,5)$, because $(1,2),(5,6)$ are red.
These are the two crossing-matching structures shown in
\Cref{fig:eightv_c2}; the derivation applies to any colour $i$, including brown.

\begin{figure}[htbp]
    \centering
\centering    
\begin{subfigure}[b]{0.3\textwidth}
         \centering
         \fbox{\includegraphics[width=\dimexpr\linewidth-2\fboxsep-2\fboxrule\relax]{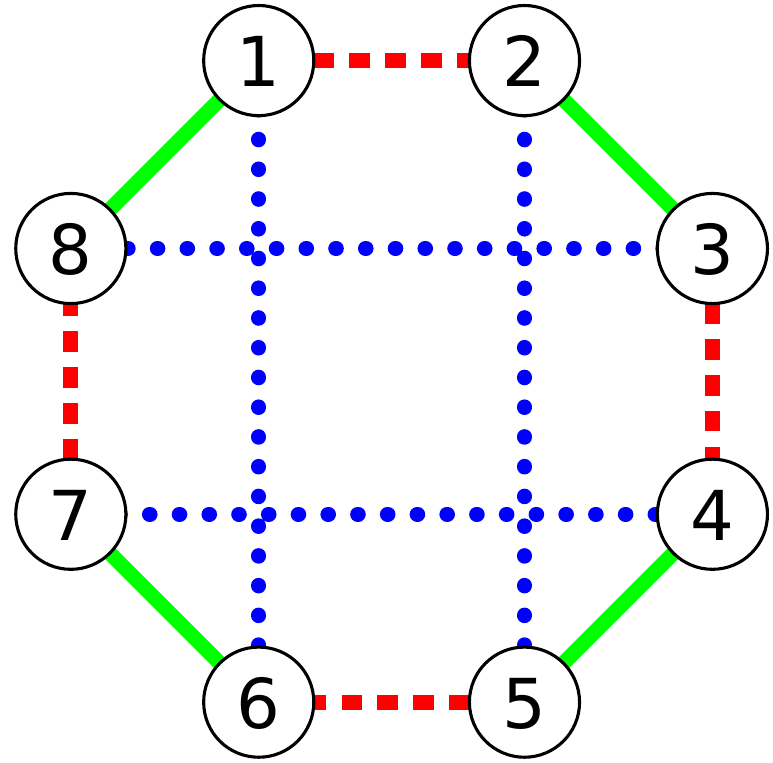}}
         \caption{}
         \label{fig:c21}
\end{subfigure}
\hfill
\begin{subfigure}[b]{0.3\textwidth}
         \centering
         \fbox{\includegraphics[width=\dimexpr\linewidth-2\fboxsep-2\fboxrule\relax]{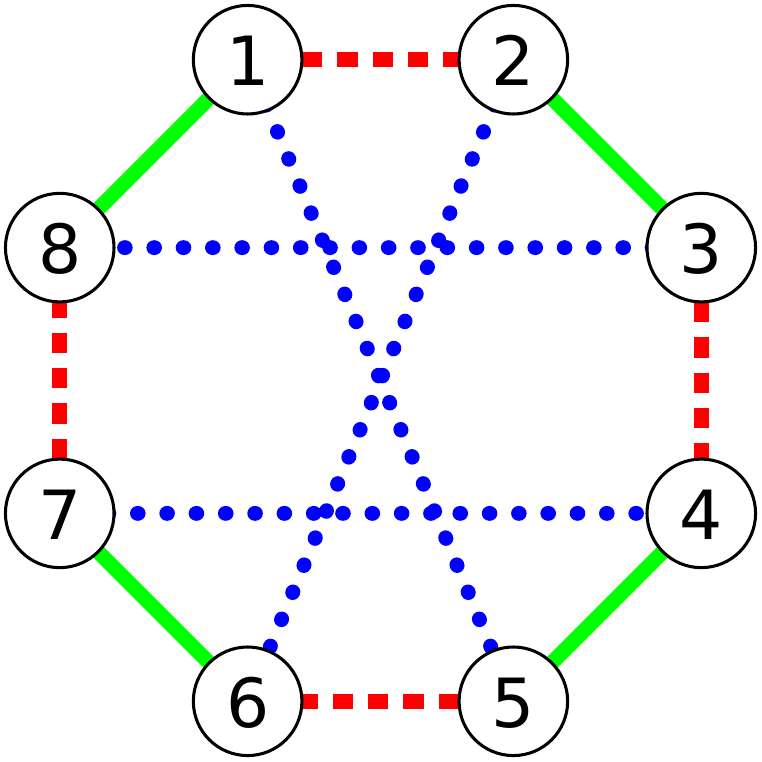}}
         \caption{ }
         \label{fig:c22}
\end{subfigure}

\caption{The two crossing-matching structures with three colours.}
\label{fig:eightv_c2}
\end{figure}

\textbf{Case 2.1.3.} This is the case shown in \cref{fig:c13}. Now consider the fourth colour, which is brown. 

\textbf{Case 2.1.3.1} If brown is not crossing, no brown monochromatic perfect matching contains a crossing edge. The union of any selected brown monochromatic perfect matching with the fixed red-green Hamiltonian cycle must then be one of the three structures from \cref{fig:eightv_c1}. Further, as all antipodal pairs have a blue edge between them, the union of a brown perfect matching with a Hamiltonian cycle with alternating red and green colours must be like that of \cref{fig:c11}, which was already proven to be not possible.

\textbf{Case 2.1.3.2} If brown is crossing, select a brown monochromatic perfect matching containing a crossing edge. By the crossing-edge derivation above, its union with the fixed red-green Hamiltonian cycle must be one of the two structures from \cref{fig:eightv_c2}. Further, as all antipodal pairs have a blue edge between them, the union of a brown perfect matching with a Hamiltonian cycle with alternating red and green colours must be like that of \cref{fig:c21}. Now consider the non-monochromatic vertex colouring in which $1,2,5,6$ get the colour brown, and $3,4,7,8$ get the colour blue. From \cref{colour_split_lemma} and \cref{partition_split_lemma}, we know that the weight of such a vertex colouring is $w((1,6))w((2,5))w((3,7))w((4,8))\ne 0.$ Contradiction.

\textbf{Case 2.2.} In this case, every non-redundant colour other than red and green is crossing.
Select a blue monochromatic perfect matching and a brown monochromatic
perfect matching, each containing a crossing edge.
The crossing-edge derivation above applies separately to both selected
matchings, so each has one of the two structures in \Cref{fig:eightv_c2},
up to a symmetry of the fixed red-green Hamiltonian cycle.
Their combined configurations are shown in \Cref{fig:eightv_c4}.

\begin{figure}[htbp]
    \centering
\centering    
\begin{subfigure}[b]{0.3\textwidth}
         \centering
         \fbox{\includegraphics[width=\dimexpr\linewidth-2\fboxsep-2\fboxrule\relax]{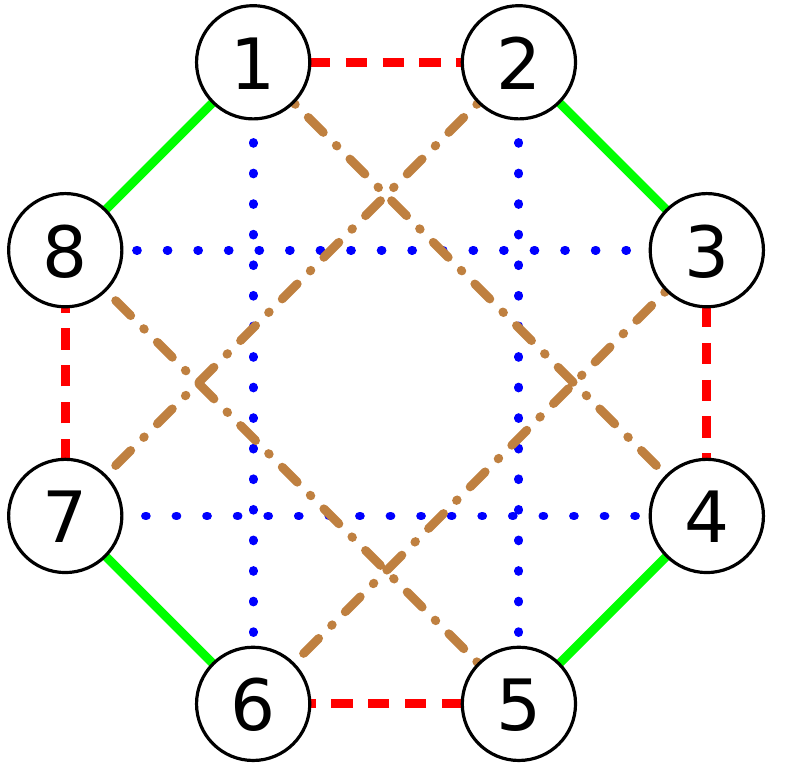}}
         \caption{Case 2.2.1}
         \label{fig:c41}
\end{subfigure}
\hfill
\begin{subfigure}[b]{0.3\textwidth}
         \centering
         \fbox{\includegraphics[width=35mm]{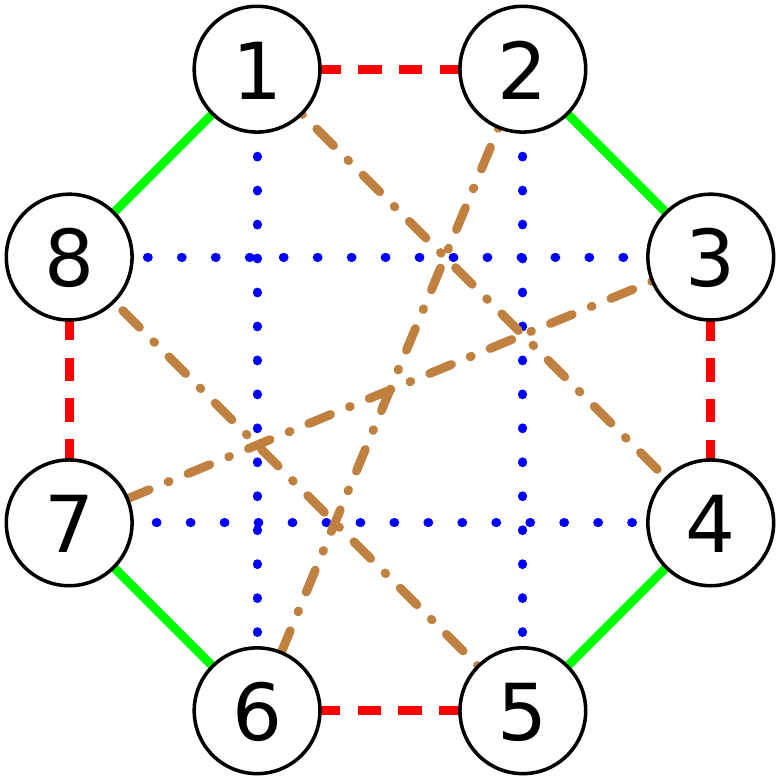}}
         \caption{Case 2.2.2 }
         \label{fig:c42}
\end{subfigure}
\hfill
\begin{subfigure}[b]{0.3\textwidth}
         \centering
         \fbox{\includegraphics[width=\dimexpr\linewidth-2\fboxsep-2\fboxrule\relax]{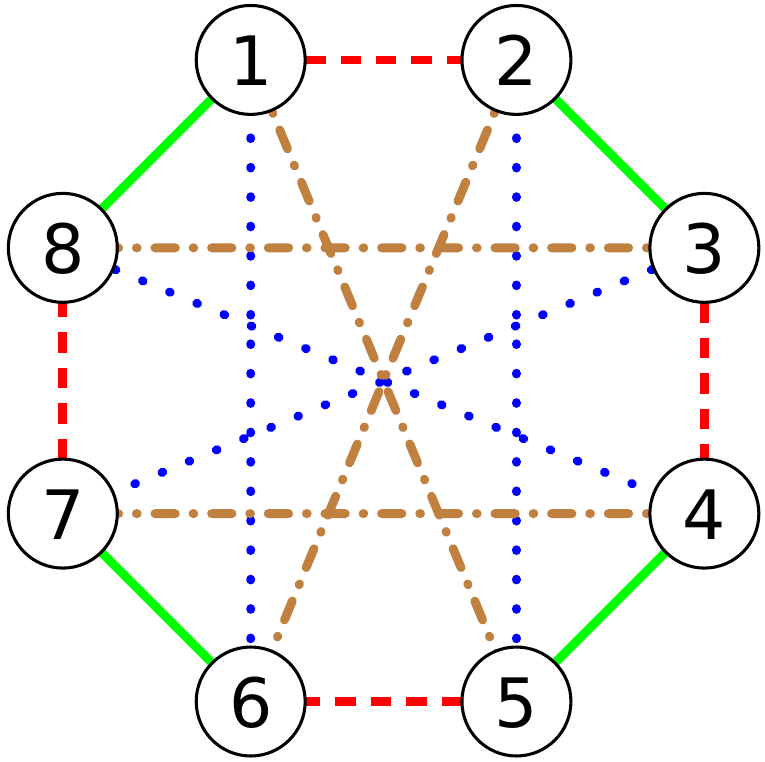}}
         \caption{Case 2.2.3}
         \label{fig:c43}
\end{subfigure}

\caption{Case 2.2 with $4$ colours}
\label{fig:eightv_c4}
\end{figure}

The selected blue and brown matchings are edge-disjoint because the graph
is simple. Up to the symmetries of the
alternating Hamiltonian cycle and interchange of the names blue and brown, the
possible pairs of blue and brown matchings are exactly the three configurations
shown in \cref{fig:eightv_c4}.

\textbf{Case 2.2.1 and Case 2.2.2} (\cref{fig:c41} and \cref{fig:c42}) {In both configurations, the following argument applies.  Consider the non-monochromatic vertex colouring in which the vertices $5,6,7,8$ are red, $1,4$ are brown, and $2,3$ are green. From \cref{colour_split_lemma}, this is possible only if the edges $(5,7)$ and $(6,8)$ are red. Now consider the non-monochromatic vertex colouring in which vertices $4,5,6,7$ are green, vertices $1,2$ are red, and vertices $8,3$ are blue. From \cref{colour_split_lemma}, this is possible only if the edges $(5,7)$ and $(6,4)$ are green. But $(5,7)$ is already red. Contradiction.}

\textbf{Case 2.2.3} (\cref{fig:c43}) Consider the vertex colouring in which $1,2,5,6$ get the colour red and $3,4,7,8$ get the colour blue. From \cref{colour_split_lemma} and \cref{partition_split_lemma}, the weight of the vertex colouring is $w((1,2))w((5,6))w((3,7))w((4,8))\neq 0$. Contradiction.

Therefore, $\mu(G) \leq 3$.

\end{proof}

\clearpage
\bibliography{references}

\end{document}